\documentclass[11pt]{amsart}

\usepackage{amsmath,amssymb}
\usepackage{graphicx}
\usepackage[colorlinks = true, linkcolor = blue, urlcolor  = blue, citecolor = red]{hyperref}

\usepackage[margin=1in]{geometry}

\newtheorem{theorem}{Theorem}[section]
\newtheorem{definition}[theorem]{Definition}
\newtheorem{proposition}[theorem]{Proposition}
\newtheorem{corollary}[theorem]{Corollary}
\newtheorem{lemma}[theorem]{Lemma}
\newtheorem{remark}[theorem]{Remark}

\newtheorem{question}[theorem]{Question}

\title[Random and free positive maps with applications to entanglement detection]{Random and free positive maps with applications to entanglement detection}
\author {Beno\^\i{}t Collins}
\address{
Department of Mathematics, Graduate School of Science, Kyoto University, Kyoto 606-8502, Japan
and
D\'epartement de Math\'ematique et Statistique, Universit\'e d'Ottawa,
585 King Edward, Ottawa, ON, K1N6N5 Canada, 
CNRS, Institut Camille Jordan Universit\'e  Lyon 1, 43 Bd du 11 Novembre 1918, 69622 Villeurbanne
France
} 
\email{collins@math.kyoto-u.ac.jp}
\author {Patrick Hayden}
\address{Physics Department, Stanford University. Varian Physics Building, 382 Via Pueblo Mall, Stanford, 94305-4060, USA.} \email{phayden@stanford.edu}
\author{Ion Nechita}
\address{Zentrum Mathematik, M5, Technische Universit\"at M\"unchen, Boltzmannstrasse 3, 85748 Garching, Germany, and CNRS, Laboratoire de Physique Th\'eorique, Toulouse, France}
\email{nechita@irsamc.ups-tlse.fr}

\date{\today}

\begin{document}

\begin{abstract}
We apply random matrix and free probability techniques to the study of linear maps of interest in quantum information theory. Random quantum channels have already been widely investigated with
spectacular success. Here, we are interested in more general maps, asking only for $k$-positivity instead of the complete positivity required of quantum channels.
Unlike the theory of completely positive maps, the theory of $k$-positive maps is far from being 
completely understood, and our techniques give many new parametrized families of such maps.

We also establish a conceptual link with free probability theory, and show that our constructions
can be obtained to some extent without random techniques in the setup of free products of von Neumann algebras. 

Finally, we study the properties of our examples and show that for some parameters, they are indecomposable. In particular, they can be used to detect the presence of entanglement missed by the partial transposition test, that is, PPT entanglement.
As an application, we considerably refine our understanding of PPT states in the case where one of the spaces
is large whereas the other one remains small. 

\end{abstract}

\maketitle

\tableofcontents

\section{Introduction}

Completely positive maps play a privileged role in quantum information theory because they provide the quantum mechanical generalization of the notion of a noisy channel. The Stinespring theorem provides a useful classification of these maps that also justifies their physical relevance as ignorance-induced noisy summaries of the unitary evolution on a larger system~\cite{stinespring}. Random methods have played an important role in the study of noisy quantum channels, notably in the proofs of capacity theorems~\cite{sw,how,hhwy} and in the construction of channels with unusual properties~\cite{hlsw,ss,fhs}, including channels violating the additivity conjecture~\cite{hw,hastings,cn1,bcn13}.

Maps with weaker positivity properties are of interest in quantum information theory precisely because they fail in subtle ways to be physically realizable. Most importantly, positive but not completely positive maps acting on entangled states may fail to preserve positivity, mapping states to non-positive operators. Since such maps will always preserve positivity for separable quantum states, which are by definition only classically correlated, positive but not completely positive maps can be used to detect the presence of entanglement. The most famous such test is partial transposition; the states with no entanglement detectable this way are known as the PPT states \cite{per}. 
In this article, we use free probability and random matrix techniques to construct new families of positive but not completely positive maps and apply these maps to the study of entanglement.

As an illustration of the power of the method, we consider quantum states of a bipartite space of finite dimension. Motivated originally by debates about whether room temperature nuclear magnetic resonance experiments were capable of producing entanglement, researchers determined some time ago that there is a neighborhood of separable states in the vicinity of the maximally mixed state~\cite{bcjlp}, and have produced successively sharper estimates on the size of that neighborhood \cite{gba}. We consider the related question of perturbing the maximally mixed state by adding to it a random Hermitian traceless matrix. This question has been studied in turn by Aubrun, Szarek, Ye, and others \cite{aub,asy,bne12}. In the limit where at least one of the systems has large dimension, one finds sharp transitions between separable states, PPT states and, finally, states whose entanglement can be detected using partial transposition. We construct new positive maps capable of detecting entanglement to which the partial transposition test is insensitive.

Formally, let $\mathcal A, \mathcal B$ be two $C^*$-algebras. Positive elements in a $C^*$ algebra are elements $x$ that are self-adjoint and that can be written
as $x=yy^*$ for some $y$ in the $C^*$-algebra. 
A positive linear map $\Phi: \mathcal A\to \mathcal B$ is a  map that sends a positive element to a positive element. 
It is well-known that if $\Phi$ is positive, then $\Phi_k=\mathrm{id}_k\otimes \Phi: \mathcal{M}_k(\mathbb{C})\otimes \mathcal A\to \mathcal{M}_k(\mathbb{C})\otimes \mathcal B$ 
is not necessarily positive, unless $\mathcal A$ or $\mathcal B$ are commutative. 
For example if $\mathcal A=\mathcal B=\mathcal{M}_n(\mathbb{C}), n\geq 2$ and $\Phi$ is the transpose map, then, $\Phi_k$ is not positive as soon as $k\geq 2$.

However if $\Phi_k$ is positive, then for $l\leq k$, the map  $\Phi_l$ is clearly positive too. 
A map $\Phi$ such that  $\Phi_k$ is positive is said to be $k$-positive. A completely positive map is one that has this property for all integers $k$. 
Note that if $\mathcal A=\mathcal{M}_n(\mathbb{C})$ and $\mathcal B=\mathcal M_d(\mathbb C)$, 
$\min (n,d)$-positive is equivalent to completely positive \cite[Corollary 4.19]{sto}.
There are many examples of maps that are $k$ positive but not $(k+1)$-positive, but the classification of positive, or even $k$-positive maps 
is far from complete. 

One of the aims of this paper is to provide a new systematic method to obtain $k$-positive maps from 
$\mathcal{M}_n(\mathbb{C})\to \mathcal{M}_d(\mathbb{C})$, based on 
free probability techniques. More precisely, for all integers $n, k$ such that $n\geq k$, we describe  infinite families
of maps $\Phi: \mathcal{M}_n(\mathbb{C})\to \mathcal{M}_d(\mathbb{C})$ that are $k$-positive but not $(k+1)$-positive. 
Our families are defined for any $d$ large enough (and the threshold depends on $n,k$), and their positivity properties depend on the behaviour of the support of the free convolution powers of a given compactly supported probability measure.

Our first main result is Theorem \ref{thm:k-positivity-from-mu}, and can be stated as follows (we refer to 
Section \ref{sec:free-prob-version} for details):

\begin{theorem}
In a non-commutative probability space where a copy of the $n\times n$ matrices $\mathcal M_n(\mathbb C)$ is free from a self-adjoint element
$X$ of distribution $\mu$, consider the map $\Phi_\mu : \mathcal{M}_n(\mathbb{C})\to \mathcal M$ given by 
$\Phi_\mu (E_{ij})=E_{1i}XE_{j1}$, with $E = (E_{ij})_{i,j\in\{1,\ldots,n\}}$ the basis of matrix units.
Then, the map $\Phi_\mu$ is $k$-positive if and only if $\mathrm{supp}(\mu^{\boxplus n/k})\subseteq [0, \infty)$.
\end{theorem}

The second main series of results looks into a random matrix approximation of the above theorem, and focuses
on the Gaussian setup. We summarize the results below, and refer to Section \ref{sec:GUE} for details.

\begin{theorem}
Let $Z_d$ be a shifted GUE random matrix in $\mathcal M_n(\mathbb C) \otimes \mathcal M_d(\mathbb C)$ of mean $2$ and variance $\alpha \in [0,1)$. Then there is a linear map $\Phi_d:\mathcal M_n(\mathbb C) \to \mathcal M_n(\mathbb C)$ such that its Choi matrix is again a shifted GUE matrix (which depends on $Z_d$), of mean $(2+\varepsilon)/\sqrt{n}$ and variance $1$.  With probability one as $d\to \infty$:
\begin{enumerate}
\item The matrix $Z_d$ is positive and PPT.
\item The map $\Phi_d$ is positive. However, $\Phi_d$ is not completely positive, provided that $2+\varepsilon< \sqrt n$.
\item As soon as $2(2+\varepsilon)<\alpha \sqrt n$, the map $\Phi_d$ detects the entanglement present in $Z_d$.
\end{enumerate}
\end{theorem}

Our paper is organized as follows. 
Section \ref{sec:positive} sets up some notation about positive and $k$-positive maps, and recalls some results about their Choi matrices. 
Section \ref{sec:free-prob-version} is about free probability: it recalls abstract facts about freeness, the 
free additive convolution semigroup, and it gives our main new construction of $k$-positive maps.
Section \ref{sec:RMT} is about the finite dimensional random approximation of the model of Section \ref{sec:free-prob-version}, and
Section \ref{sec:GUE} specifies to the case of non-centered $\mathrm{GUE}$ random matrices, with applications to typical radii of convergence for PPT vs.~ separability, along with non-decomposability results for our random maps. 

\bigskip

\noindent \textit{Acknowledgements.}
B.C.'s research was partly supported by NSERC, ERA, and Kakenhi funding. 
P.H.'s research is supported by CIFAR, FQXi and the Simons Foundation.
I.N.'s research has been supported by a von Humboldt fellowship and by the ANR projects {OSQPI} {2011 BS01 008 01} and {RMTQIT}  {ANR-12-IS01-0001-01}. Both B.C. and I.N. were supported by the ANR project {STOQ}  {ANR-14-CE25-0003}.

\section{Positive maps: notations and facts}
\label{sec:positive}

We denote by $\mathcal{M}_n(\mathbb{C})$ the algebra of $n\times n$ complex matrices. Let $E=(E_{ij})_{i,j\in \{1,\ldots n\}}$ be the basis of matrix units,
i.e. $E_{ij}E_{kl}=E_{il}\delta_{jk}$ and $E_{ij}^*=E_{ji}$. For a  $C^*$-algebra $\mathcal A$, let $\Phi:\mathcal{M}_n(\mathbb{C})\to \mathcal A$ be a linear map. Its Choi matrix in the basis $E$, denoted by $C_\Phi$ is an element of $\mathcal{M}_n(\mathbb{C})\otimes \mathcal A$ defined as
$$C_\Phi=\sum_{i,j=1}^n E_{ij}\otimes \Phi(E_{ij}).$$
Denoting by $B_d \in \mathbb C^d \otimes \mathbb C^d$ the normalized Bell state
\begin{equation}\label{eq:B-d}
B_d = \frac{1}{\sqrt d} \sum_{i=1}^d e_i \otimes e_i,
\end{equation}
the Choi matrix can also be written as 
$$C_\Phi = d [\mathrm{id}_d \otimes \Phi](B_d B_d^*).$$

A map $\Phi : \mathcal M_n(\mathbb C) \to \mathcal A$ is called \emph{$k$-positive} if the dilated application 
$$\Phi \otimes \mathrm{id}_k : \mathcal M_n(\mathbb C) \otimes \mathcal M_k(\mathbb C) \to \mathcal A  \otimes \mathcal M_k(\mathbb C)$$
is positive. Moreover, such a map is called \emph{completely positive} if it is $k$-positive for all $k$. According to a celebrated result of  Choi \cite{cho}, $\Phi$ is completely positive if and only if $C_\Phi$ is positive in $\mathcal{M}_n(\mathbb{C}) \otimes \mathcal A$. A characterization of $k$-positivity using the Choi-Jamio{\l}kowski isomorphism has been obtained in \cite{tho}. 

We recall the following result from \cite{hlp+}, characterizing $k$-positivity of maps. 
\begin{proposition}{\cite[Proposition 2.2]{hlp+}}\label{prop:k-positive-from-Choi}
Consider a linear map $\Phi : \mathcal M_n(\mathbb C) \to \mathcal{B}(\mathcal H)$, where $\mathcal H$ is a Hilbert space. The following are equivalent:
\begin{enumerate}
\item The map $\Phi$ is $k$-positive.
\item For all vectors $x = \sum_{i=1}^k y_i \otimes z_i$, with $y_i \otimes z_i \in \mathbb C^n \otimes \mathcal H$, 
$\langle x, C_\Phi x \rangle \geq 0$.
\item The operator $(I_n \otimes P) C_\Phi (I_n \otimes P)$ is positive semidefinite for any rank $k$ orthogonal projection $P \in \mathcal{B}(\mathcal H)$.
\item The operator $(P \otimes 1_\mathcal H ) C_\Phi (P \otimes 1_\mathcal H)$ is positive semidefinite for any rank $k$ orthogonal projection $P \in \mathcal M_n(\mathbb C)$  ($1_\mathcal H$ is the identity operator in $\mathcal{B}(\mathcal H)$).
\end{enumerate}
\end{proposition}

\section{A model via free probability}
\label{sec:free-prob-version}

We start this section by recalling some facts from free probability theory that will be used throughout the paper. We then introduce our construction of linear maps from probability measures and we characterize their block-positivity properties (random matrix versions of these results will be considered in Section \ref{sec:RMT}). Finally, we discuss some simple, but illuminating examples in the final subsection. 

\subsection{Freeness and free products}
The monographs \cite{vdn} and \cite{nsp} are excellent introductions to free probability theory; here, we recall only the concepts needed in this paper. 

A \emph{$*$-non-commutative probability space} is
a unital $*$-algebra $\mathcal A$ endowed with a tracial 
state $\phi$, i.e. a linear map $\phi\colon\mathcal A\to\mathbb C$
satisfying $\phi (ab)=\phi (ba),\phi (aa^{*})\geq 0, \phi (1)=1$.
An element of $\mathcal A$ is called
a \emph{non-commutative random variable}. 

Let $\mathcal A_1, \ldots ,\mathcal A_k$ be subalgebras of $\mathcal A$ having the same unit as $\mathcal A$.
They are said to be \emph{free} if for all $a_i\in \mathcal  A_{j_i}$ ($i=1, \ldots, k$) 
such that $\phi(a_i)=0$, one has  
$$\phi(a_1\cdots a_k)=0$$
as soon as $j_1\neq j_2$, $j_2\neq j_3,\ldots ,j_{k-1}\neq j_k$.
Collections $S_{1},S_{2},\ldots $ of random variables are said to be 
free if the unital subalgebras they generate are free.

Let $(a_1,\ldots ,a_k)$ be a $k$-tuple of self-adjoint random variables and let
$\mathbb{C}\langle X_1 , \ldots , X_k \rangle$ be the
free $*$-algebra of non-commutative polynomials on $\mathbb{C}$ generated by
the $k$ self-adjoint indeterminates $X_1, \ldots ,X_k$. 

The {\it joint distribution\it} of the family $\{a_i\}_{i=1}^k$ is the linear form
\begin{align*}
\mu_{(a_1,\ldots ,a_k)} : \mathbb C \langle X_1, \ldots ,X_k \rangle &\to \mathbb C \\
P &\mapsto \phi (P(a_1,\ldots ,a_k)).
\end{align*}

Given a $k$-tuple $(a_1,\ldots ,a_k)$ of free 
random variables such that the distribution of $a_i$ is $\mu_{a_i}$, the joint distribution
$\mu_{(a_1,\ldots ,a_k)}$ is uniquely determined by the
$\mu_{a_i}$'s.

In particular, $\mu_{a_1+a_2}$ and $\mu_{a_1a_2}$ depend only on
$\mu_{a_1}$ and $\mu_{a_2}$. The notations $\mu_{a_1+a_2}=
\mu_{a_1}\boxplus\mu_{a_2}$ and $\mu_{a_1a_2}=\mu_{a_1}\boxtimes
\mu_{a_2}$ were introduced in Voiculescu's works \cite{voi86,voi87}; the operations $\boxplus$ and $\boxtimes$ 
are called the {\em free additive}, respectively
{\em free multiplicative} convolution. Moreover, for any compactly supported probability measure $\mu$, one can define its free additive convolution  powers $\mu^{\boxplus T}$ for any real parameter $T \geq 1$, see \cite[Corollary 14.13]{nsp}.

The following result will be quite useful for us and was proved by Nica and Speicher \cite[Exercise 14.21]{nsp}.

\begin{lemma}\label{lem:nica-speicher}
Let $a,p$ be free elements in a non-commutative probability space $(\mathcal A,\phi)$ and assume that $p$ is a self-adjoint projection of rank $t \in (0,1)$ and that $a$ is a self-adjoint random variable having distribution $\mu$.
Then, the distribution of $t^{-1} p a p$ in $(p \mathcal A p , \phi (p \cdot ))$ is $\mu^{\boxplus 1/t}$.
\end{lemma}

Finally, we discuss free products.
Given $(\mathcal A_i,\phi_i)_{i\in I}$ a family of non-commutative probability spaces, there exists a non-commutative probability space $(\mathcal A,\phi)$ in which all $(\mathcal A_i,\phi_i)$ can be embedded in 
a trace-preserving way such that they are free with respect to each other. This is called Voiculescu's \emph{free product construction}, and we will need it below.
There exists a $*$-algebra version of this construction, a $C^*$-algebra and a von Neumann algebra construction.
If $\mathcal A$ is generated by the $\mathcal A_i$'s, we write $(\mathcal A,\phi)=*_{i\in I}(\mathcal A_i,\phi_i)$.

\subsection{Maps associated to probability measures}
We now introduce the main idea of this paper, a construction of a linear map depending on a probability measure $\mu$. 
We start with a compactly supported probability measure $\mu$ on $\mathbb{R}$ and fix an integer $n$. 
We will be interested in the space $L^\infty (\mathbb{R},\mu)$. Note that this space is spanned as a von Neumann algebra
 by the operator $x \mapsto x$, which we will denote by $X$. By construction, $X$ is a self-adjoint operator and its spectrum is the support of $\mu$.

We consider the von Neumann probability space free product \cite{vdn}
$$( \tilde{\mathcal M}, \mathrm{tr}*\mathbb E):= (\mathcal{M}_n(\mathbb{C}),\mathrm{tr})*(L^\infty (\mathbb{R},\mu), \mathbb E)$$
and the contracted von Neumann probability space $(\mathcal M,\tau )$ where
$\mathcal M=E_{11}\tilde{\mathcal M} E_{11}$ is the contracted algebra of $\tilde{\mathcal M}$ and $\tau$ is the restriction of $\mathrm{tr} * \mathbb E$ appropriately normalized (by a factor $n$).

We consider the map $\Phi_\mu : \mathcal{M}_n(\mathbb{C})\to \mathcal M$ given by 
$$\Phi_\mu (E_{ij})=E_{1i}XE_{j1}.$$

Our main result is as follows:

\begin{theorem}\label{thm:k-positivity-from-mu}
Let $\mu$ be a compactly supported probability measure. The map $\Phi_\mu$ defined above is $k$-positive if and only if $\mathrm{supp}(\mu^{\boxplus n/k})\subseteq [0, \infty)$.
\end{theorem}
\begin{proof}
The map $\Phi_\mu$ was constructed in such a way that its Choi matrix has distribution $\mu$. Indeed, the von Neumann algebra $\mathcal M$ is a corner of $\tilde{\mathcal M}$, so the map $a \mapsto \sum_{i,j=1}^n E_{ij} \otimes E_{1i} a E_{j1}$ implements an isomorphism between $\tilde{\mathcal M}$ and $\mathcal M_n(\mathcal M)$; hence, $C_{\Phi_\mu} = X$. For any rank $k$ orthogonal projection $P \in \mathcal M_n(\mathbb C)$, the distribution of the element $(P \otimes 1_{\mathcal M}) C_{\Phi_\mu} (P \otimes 1_{\mathcal M})$ is $[(1-k/n) \delta_0 + k/n \delta_1] \boxtimes \mu$. By the Nica-Speicher result \ref{lem:nica-speicher} above, this distribution has positive support if and only if $\mathrm{supp}(\mu^{\boxplus n/k})\subseteq [0, \infty)$. The result follows now from an application of Proposition \ref{prop:k-positive-from-Choi}.
\end{proof}

The remainder of the paper is devoted to the study of the maps $\Phi_\mu$ and to applications of the above construction to Quantum Information Theory. 

\subsection{Examples}

Let us start the subsection devoted to examples by making some general purpose remarks. First, note that in order for the condition in Theorem \ref{thm:k-positivity-from-mu} to hold for some value of $k$, the measure $\mu$ must have non-negative average. Indeed, the average of $\mu^{\boxplus n/k}$ is $n/k$ times the average of $\mu$. It follows that, if one starts with a probability distribution $\mu$ having negative mean,  the measures $\mu^{\boxplus n/k}$ cannot have support in $[0, \infty)$. In particular, in order to obtain non-trivial examples, we have to rule out any distributions which are symmetric with respect to the origin.

Similarly, a measure $\mu$ whose support is already in $\mathbb{R}_+$ will not be interesting for us because the corresponding map $\Phi_\mu$ will automatically be completely positive and thus $k$-positive for all $1 \leq k \leq n$ (recall that we are ultimately interested in positive maps which are not completely positive).
In the following, we are interested in \emph{non-trivial} examples of measures $\mu$, i.e. measures such that
$$\mathrm{supp}(\mu) \nsubseteq [0, \infty) \quad \text{but} \quad \mathrm{supp}(\mu^{\boxplus T}) \subseteq [0, \infty) \text{ for some $T \geq 1$.}$$

\subsubsection{The semicircular case}
We start with semicircular distributions, which will play a central role in Section \ref{sec:GUE}. As mentioned above, in order to produce non-trivial examples of maps $\Phi_\mu$, one should rule out symmetric distributions $\mu$; we consider next shifted (non-centered) semicircle distributions. Let us consider $\mathrm{SC}_{a,\sigma}$ the Wigner semicircle distribution of mean $a$ and variance $\sigma^2$:
$$\mathrm{SC}_{a,\sigma} = \frac{\sqrt{4\sigma^2-(x-a)^2}}{2\pi \sigma^2} \mathbf{1}_{[a-2\sigma ,a+ 2\sigma]}(x) dx.$$
Note that the support of the semicircle measure is $[a-2\sigma,a+2\sigma]$. The semicircle distribution plays a role in free probability theory analogous to the role played by the Gaussian distribution in classical probability. Its free cumulants are all zero, except for the first two, which read $\kappa_1 = a$, $\kappa_2  = \sigma^2$. Since the free cumulants linearize the free additive convolution \cite[Proposition 12.3]{nsp}, their values for $\mathrm{SC}_{a,\sigma}^{\boxplus n/k}$ are $\kappa_1 = an/k$, respectively $\kappa_2=\sigma^2n/k$, all the others being zero. Hence, $\mathrm{SC}_{a,\sigma}^{\boxplus n/k}=\mathrm{SC}_{an/k,\sigma\sqrt{n/k}}$, and thus the support of $\mathrm{SC}_{a,\sigma}^{\boxplus n/k}$ is
$$\mathrm{supp}\left( \mathrm{SC}_{a,\sigma}^{\boxplus n/k} \right) = \left[\frac{an}{k} - 2\sigma\sqrt{\frac n k},\frac{an}{k} + 2\sigma\sqrt{\frac n k} \right].$$

In particular, we have:

\begin{proposition}
Let $n$ be an integer, $a \in \mathbb R$, and $\sigma>0$, and consider the map $\Phi^{(a,\sigma)}_{\mathrm{SC}}: \mathcal{M}_n(\mathbb{C})\to \mathcal M$ associated to the semicircle distribution $\mathrm{SC}_{a,\sigma}$. The map $\Phi^{(a,\sigma)}_{\mathrm{SC}}$ is $k$-positive if and only if $k\leq a^2n/(4\sigma^2)$.
In particular, for any $n$ and any $1 \leq k < n$, there exist parameters $a,\sigma>0$ such that the above map is $k$-positive but not $k+1$-positive.
\end{proposition}

\subsubsection{The Marchenko-Pastur case}\label{subsubMP}

Here we consider $m_{a,t}$ to be the distribution of the random variable $1-aX_t$ where $X_t$ has a Marcenko-Pastur (or free Poisson of parameter $t$)
distribution and $a>0$.
We recall that the free Poisson law of parameter $t$ is given by (see, e.g.~ \cite[Definition 12.12]{nsp}):
$$\pi_t=\max (1-t,0)\delta_0+\frac{\sqrt{4t-(x-1-t)^2}}{2\pi x}1_{[1+t-2\sqrt{t},1+t+2\sqrt{t}]}dx,$$
and that it has the semigroup property 
$\pi_s\boxplus\pi_t=\pi_{s+t}$ (this follows from the fact that the free cumulants of $\pi_t$ read $\kappa_p = t$, $\forall p \geq 1$). In particular, this implies for our purposes that the bottom of the spectrum of
$m_{a,t}$ is located at
$1-a(1+t+2\sqrt{t})$. It follows from Nica and Speicher's Lemma \ref{lem:nica-speicher} that if $p$ is a free projection of rank $k/n$, the bottom of the spectrum of the 
operator $p(1-aX_t)p$ is located at
$$1-\frac{k}{n}a \left( 1+t\frac{n}{k}+2\sqrt{tn/k} \right).$$
We summarize what we obtained above as follows. 

\begin{proposition}\label{prop:mp}
Let $n$ be an integer, $a,t>0$, $\Phi^{(t)}_{\mathrm{MP}}: \mathcal{M}_n(\mathbb{C})\to \mathcal M$ whose Choi map is associated to the non-commutative random variable
$1-aX_t$.
If $1-\frac{k}{n}a(1+t\frac{n}{k}+2\sqrt{tn/k})\geq 0$ but $1-\frac{k+1}{n}a(1+t\frac{n}{k+1}+2\sqrt{tn/(k+1)})< 0$ then
the map $\Phi^{(t)}_{\mathrm{MP}}$ is $k$-positive but not $(k+1)$-positive. Moreover, for each $n,k$ there exists a choice of $a,t$ that satisfies the above inequalities.
\end{proposition}

\subsubsection{The `small rank projection' case}\label{subsub:small-rank}

Now we are interested in the case where $X=1-aP$ where $P$ is a projection of rank $\varepsilon$ small (i.e.~ $\tau(P) = \varepsilon \in (0,1)$). 

\begin{proposition}
As $\varepsilon\to 0$, the threshold in $a$ to obtain $k$-positivity converges to $n/k$.
\end{proposition}

\begin{proof}
Let us consider some fixed $\varepsilon >0$. Then, there exists a variable $Y$ having a Marchenko-Pastur distribution of parameter $\varepsilon$ such that, for any $\eta \geq 6\sqrt\varepsilon$,
$$1-(a+\eta )Y\leq X\leq 1-(a-\eta )Y.$$
Indeed, recall that the Marchenko-Pastur distribution of parameter $\varepsilon<1$ can be written as 
$$\pi_\varepsilon = (1-\varepsilon) \delta_0  + \varepsilon \tilde \pi_\varepsilon,$$
where $ \tilde \pi_\varepsilon$ is a probability measure supported on the compact interval $[(1-\sqrt \varepsilon)^2, (1+\sqrt \varepsilon)^2]$. The result follows after applying Proposition \ref{prop:mp} and letting first $\varepsilon\to 0$ and
then $\eta\to 0$.
\end{proof}

\subsection{Quantum states detected by ``free'' maps}

Let $\mathcal D_n$ be the set of $n$-dimensional quantum states
$$\mathcal D_n := \{\rho \in \mathcal M_n \, : \, \rho \geq 0 \text{ and } \mathrm{Tr} \rho = 1\},$$  
and consider the set of \emph{separable} states
$$\mathcal{SEP}_{n,m} = \left\{ \sum_i t_i \rho_i \otimes \sigma_i \, : \, t_i \geq 0, \rho_i \in \mathcal D_n ,\sigma_i \in \mathcal D_m, \sum_i t_i=1 \right\} \subseteq \mathcal D_{nm}.$$
Given a positive map $\Phi : \mathcal M_n \to \mathcal M_d$, we introduce the convex body of $n\cdot m$-dimensional quantum states which remain positive under the partial action of $\Phi$
$$K_{\Phi, m}=\{ \rho \in \mathcal D_{nm} \, : \, [\Phi\otimes \mathrm{id}_m] (\rho) \geq 0\}.$$
We first record the following elementary result:
\begin{lemma}
For any positive map $\Phi$ and any dimension $m$, 
the set $K_{\Phi, m}$ is a convex body and it contains the set $\mathcal{SEP}_{n,m}$ of separable states. 
If $\Phi$ is \emph{not} $m$-positive, then $K_{\Phi, m}$ is a strict subset of $\mathcal D_{nm}$. 
\end{lemma}

Let us turn our attention to the particular case of a map $\Phi_\mu$ as defined in the previous sections. We start with the following proposition, showing that any non-trivial $\Phi_\mu$ detects pure entangled states. 
\begin{proposition}
Let $\mu$ be a compactly supported probability measure such that $\mathrm{supp}(\mu^{\boxplus n})\subseteq [0, \infty)$, but $\mathrm{supp}(\mu^{\boxplus n/k}) \cap (-\infty,0) \neq \emptyset$, for all $2 \leq k \leq n$. Then, the set $K_{\Phi_\mu, m}$ contains all separable states but no pure entangled state. Put differently, $\Phi_\mu$ detects perfectly the entanglement of pure states. 
\end{proposition}
\begin{proof}
The claim about separable states follows from the fact that the map $\Phi_\mu$ is (1-)positive. 
Let $p$ be a rank one self-adjoint projection on $\mathcal M_n(\mathbb C)\otimes \mathcal M_k(\mathbb C)$ that is not separable (i.e.~ $p$ is a multiple of an entangled state). Let  $q \in\mathcal M_n(\mathbb C) $ be the orthogonal projection on the range of the partial trace of $p$; since $p$ is entangled, $\mathrm{rk}(q) \geq 2$. 

By direct inspection, studying the positivity of $[\Phi_\mu \otimes \mathrm{id}_m](p)$ is exactly the same problem 
as studying the positivity of $(q\otimes 1_{\mathcal M}) X (q\otimes 1_{\mathcal M})$; indeed, using a Schur map argument, one can assume without loss of generality that the eigenvalues of the partial trace of $p$ are all equal.
However, by our assumptions, the latter element does not have a positive spectrum. Therefore,  $K_{\Phi_\mu, m}$ does not contain pure entangled states.
\end{proof}

Let us now introduce two intermediate sets between the set of separable states and the set of all quantum states, corresponding to the \emph{reduction criterion} \cite{cag, hho} and the \emph{positive partial transposition} (or PPT) criterion \cite{per}:
\begin{align*}
\mathcal {RED}_{n,m} &= \{ \rho \in \mathcal D_{nm} \, : \, I_n \otimes [\mathrm{Tr}_n \otimes \mathrm{id}_m](\rho) - \rho \geq 0\}\\
\mathcal {PPT}_{n,m} &= \{ \rho \in \mathcal D_{nm} \, : \, [\top_n \otimes \mathrm{id}_m](\rho) \geq 0\},
\end{align*}
where $\top$ denotes the transposition map. It is known \cite{cag} that these sets satisfy the following chain of inclusions (which all become strict for large enough values of $n,m$)
$$\mathcal {SEP}_{n,m} \subseteq \mathcal {PPT}_{n,m} \subseteq \mathcal {RED}_{n,m} \subseteq \mathcal {D}_{nm}.$$

We show next that the maps corresponding to the ``small rank projection'' measures (see Section \ref{subsub:small-rank}) detect the same entangled states as the reduction criterion, in a suitable asymptotic regime.

\begin{theorem}\label{thm:alg-red}
In the setup of subsection \ref{subsub:small-rank}, 
with parameters $a=n$, $\varepsilon \to 0$,
the convex body
$K_{\Phi, m}$ tends to $\mathcal{RED}_{n,m}$. 
\end{theorem}

Before we prove this result, let us provide an equivalent, but more general form of the reduction criterion. 

\begin{lemma}
Let $\Phi:\mathcal M_n(\mathbb C) \to \mathcal A$ be a (non-unital) embedding of $*$-algebras, sending the matrix units $E_{ij}$ of $\mathcal M_n(\mathbb C)$ to elements $F_{ij} \in \mathcal A$. Then, for any integer $m$ and for any quantum state $\rho = \sum_{ij=1}^n E_{ij} \otimes \rho_{ij} \in \mathcal D_{nm}$, 
 $$I_n \otimes [\mathrm{Tr}_n \otimes \mathrm{id}_m](\rho) - \rho \geq 0 \iff 1_{\mathcal A}\otimes \sum_{i=1}^n \rho_{ii} - \sum_{i,j=1}^n F_{ij} \otimes \rho_{ij} \geq 0.$$
\end{lemma}

\begin{proof}[Proof of Theorem \ref{thm:alg-red}]

As in the context of section \ref{subsub:small-rank} we are looking at a Choi matrix of the form $X_\varepsilon=1-nP$ where $P$ is a projection of trace $\varepsilon$. Let $\Phi_\varepsilon$ be the corresponding map (obtained from $X_\varepsilon$ via Theorem \ref{thm:k-positivity-from-mu}).

For a given $\mathcal D_{nm} \ni \rho = \sum_{ij=1}^n E_{ij} \otimes \rho_{ij} \in \mathcal D_{nm}$, we have
$$[\Phi_\varepsilon \otimes \mathrm{id}_m](\rho) = 1_{\mathcal M} \otimes \sum_{i=1}^n \rho_{ii} - \sum_{i,j=1}^n F^{(\varepsilon)}_{ij} \otimes \rho_{ij} \geq 0,$$
where $F^{(\varepsilon)}_{ij} = n E_{1i}P E_{j1}$. Let us show now that the $F^{(\varepsilon)}_{ij}$ are \emph{approximate matrix units}, i.e.~ they verify the $*$-algebra relations for matrix units up to some vanishing precision as $\varepsilon \to 0$. To show that, as $\varepsilon \to 0$, for all $i,j,k,l$,
\begin{equation}\label{eq:ref}
\|F^{(\varepsilon)}_{ij}F^{(\varepsilon)}_{kl} - \delta_{jk}F^{(\varepsilon)}_{il} \|\to 0.
\end{equation}
we just have to compute
$$n^2E_{1i}PE_{j1}E_{1k}PE_{l1}-n\delta_{jk}E_{1i}PE_{l1}= nE_{1i} (nPE_{jk}P-\delta_{jk}P)E_{l1}.$$
We show next that, as $\varepsilon \to 0$, 
\begin{equation}\label{eq:approx-matrix-units}
\| nPE_{jk}P-\delta_{jk}P \|\to 0.
\end{equation}
Let us first consider the case $j=k$. Again from Nica and Speicher's Lemma \ref{lem:nica-speicher}, it follows that the distribution of the element $P(nE_{jj}-I)P$ in the reduced non-commutative probability space $P \mathcal M P$ is that of the dilation of the  free additive power $D_\varepsilon\left[\left(n^{-1}\delta_{n-1} + (1-n^{-1})\delta_{-1}\right)^{\boxplus 1/\varepsilon}\right]$. Let $\varepsilon = 1/N$, for some (large) integer $N$. Then, by the super-convergence result of Bercovici and Voiculescu \cite[Theorem 7]{bvo} (see also \cite[Example 1]{kar}), we have that the support of 
$$ D_{N^{-1/2}}\left[\left(n^{-1}\delta_{n-1} + (1-n^{-1})\delta_{-1}\right)^{\boxplus N}\right]$$
is contained, for $N$ large enough, into some fixed compact interval. The normalization in $N^{-1}$ shows the desired convergence in norm to zero. The case $j \neq k$ is treated in a similar manner, by considering the symmetrized version $\|P(E_{jk}  +E_{kj})P \| = 2\| PE_{jk}P\|$.

By multiplying on the left and on the right equation \eqref{eq:approx-matrix-units} by $E_{1i}$, resp. $E_{l1}$ which are both of
operator norm one, one obtains \eqref{eq:ref}; note that in all the computations above, $n$ is a fixed parameter.

It follows from \cite[Proposition II.8.3.22]{bla} that approximate matrix units can be approximated by \emph{true} (exact) matrix units: there exist true matrix units $G^{(\varepsilon)}_{ij} \in \mathcal M$ such that, for all $i,j$, $\|F^{(\varepsilon)}_{ij} - G^{(\varepsilon)}_{ij}\| \to 0$ as $\varepsilon \to 0$. Define 
$$\Psi_{\varepsilon,m} (\rho) = 1_{\mathcal M} \otimes \sum_{i=1}^n \rho_{ii} - \sum_{i,j=1}^n G^{(\varepsilon)}_{ij} \otimes \rho_{ij} \geq 0.$$
By the previous lemma, $\rho \in \mathcal{RED}_{n,m}$ iff. $\Psi_{\varepsilon,m} (\rho) \geq 0$ in $\mathcal M$. But we have
$$\| [\Phi_\varepsilon \otimes \mathrm{id}_m](\rho)  - \Psi_{\varepsilon,m} (\rho) \| = \left\|\sum_{i,j=1}^n (F^{(\varepsilon)}_{ij} - G^{(\varepsilon)}_{ij}) \otimes \rho_{ij} \right\| \leq n^2\sup_{i,j}\|F^{(\varepsilon)}_{ij} - G^{(\varepsilon)}_{ij}\|  \to 0,$$
finishing the proof. 
\end{proof}

In other words, an infinitesimally small free projection captures exactly the set of quantum states whose entanglement is detected by the reduction criterion. Since $\mathcal{RED}_{n,m} \supseteq \mathcal{PPT}_{n,m}$, an important question would is to check whether the sets $K_{\Phi_\mu, m}$ can be compared to $\mathcal{PPT}_{n,m}$ for appropriate measures $\mu$. Although we do not address this question directly, we will see with Theorem \ref{thm:GUE-PPT-entangled} that with random techniques, a result of this flavor holds.

We finish this section by a general open question. For all $n,k,l$, we introduce the set 
$$K^{free}_{n,k,m} := \bigcap_{\mu \, : \, \mathrm{supp} \left( \mu^{\boxplus n/k} \right) \subset [0, \infty)} K_{\Phi_\mu , m}.$$
We call it the collection of \emph{freely detectable quantum states} in  $\mathcal D_{nm}$.

In view of the preliminary results of this section, we believe that this set is of some interest in quantum information theory. 
The sequel of this paper will show that suitable approximations of this
set have a very strong entanglement detection power. This leads us to stating the following question, that has to our mind an interest both in free probability and quantum information theory:

\begin{question}
Can one give a description of $K^{free}_{n,k,m}$? We know this set contains the collection $k$-separable states (Theorem \ref{thm:k-positivity-from-mu}), but how much bigger is it? Are there values of the parameters $n,m,m$ for which the set $K^{free}_{n,k,m}$ is precisely the set of $k$-separable states from $\mathcal D_{nm}$?
\end{question}

\section{A finite dimensional, random version}
\label{sec:RMT}

Theorem \ref{thm:k-positivity-from-mu} and the construction of the maps $\Phi_\mu$ give a way to construct new families of positive but not completely positive maps. The results in Section \ref{sec:free-prob-version} are not very concrete in the sense that the compression of free products of von Neumann algebras is a rather abstract object. In this section, we show that it is possible to interpret the previous results as the limiting case of \emph{random} maps between finite dimensional matrix algebras. In particular, we obtain a large class of linear maps between matrix algebras that are $k$ positive but not $k+1$ positive. For this we need a result about norm approximations. 

\subsection{Norm behaviour of products of random matrices}

We recall that if $X$ is a $d$-dimensional self-adjoint matrix, its \emph{empirical eigenvalue distribution} is $d^{-1}\sum_{i=1}^{d}\delta_{\lambda_{i}}$, where $\lambda_{i}$ are the eigenvalues of $X$.
For any probability measure $\mu$ on the real line, its \emph{repartition function} is defined as $F_{\mu}: t  \mapsto \mu ((-\infty,t])$.

For the purposes of this section, we will say that a sequence of repartition functions $f_{d}$ converges to a repartition function $f$ if for all $\varepsilon >0$,
 there exists an $d_{0}$ such that for all $d\geq d_{0}$,
$$    \forall t\in\mathbb R,\quad f(t-\varepsilon )-\varepsilon \leq f_{d}(t)\leq f(t+\varepsilon ) +\varepsilon.$$

\begin{theorem}\label{thm:RMT-norm}
Let $A_{d},B_{d}$ be independent positive self-adjoint random matrices in $\mathcal{M}_{d}(\mathbb C)$, such that at least one of $A_{d}$ or $B_{d}$ has a distribution invariant under unitary conjugation. 
Let $f_{d}$ be the repartition function of $A_{d}$
and $g_{d}$ be the repartition function of $B_{d}$. Assume that the (a priori random) repartition  functions $f_{d},g_{d}$ converge almost surely respectively to $f,g$ which are repartition functions of two self-adjoint, bounded and freely independent random variables $x$ and $y$. Assume also that the operator norm of $A_d$ (resp. $B_d$) converges to the operator norm of $x$ (resp. $y$). Then, almost surely as $d \to \infty$,
$$||A_{d}B_{d}|| \to ||xy||.$$
\end{theorem}

The above result is a special case of a much more general class of results, see \cite{cm,bcn}.

\subsection{The finite dimensional, random matrix model}

In this section, we fix again a compactly supported probability measure $\mu$ and an integer $n$. For each integer $d$, we introduce a real-valued diagonal matrix $X_d \in \mathcal{M}_n(\mathbb{C})\otimes \mathcal{M}_d(\mathbb{C})$
whose empirical eigenvalue distribution converges to $\mu$ and whose extremal eigenvalues converge to the respective extrema of the support of $\mu$:
\begin{align*}
\lim_{d \to \infty} \lambda_{\min}(X_d) &= \min \{x : x \in \mathrm{supp}(\mu) \}\\
\lim_{d \to \infty} \lambda_{\max}(X_d) &= \max \{x : x \in \mathrm{supp}(\mu) \}.
\end{align*}

Let also $U_d$ be a random unitary matrix in $\mathcal{U}_{nd}$ distributed according to the Haar measure on
the unitary group. We introduce the map $\Phi_{\mu,d} : \mathcal{M}_n(\mathbb{C})\to \mathcal{M}_d(\mathbb{C})$ whose Choi matrix is the random matrix $C_d := U_dX_dU_d^*\in \mathcal{M}_n(\mathbb{C})\otimes \mathcal{M}_d(\mathbb{C})$; this can be done using the hermiticity preserving map $\leftrightarrow$ self-adjoint matrix incarnation of the Choi-Jamio{\l}kowski isomorphism. More precisely, we have
$$\forall X \in \mathcal{M}_n(\mathbb{C}), \quad \Phi_{\mu,d}(X) = [\mathrm{Tr} \otimes \mathrm{id}](C_d \cdot X^\top\otimes I_d).$$
We can now state the main result of this section.

\begin{theorem}\label{thm:k-positivity-from-mu-random-matrix}
The sequence of random linear maps $(\Phi_{\mu,d})_d$ has the following properties:
\begin{enumerate}
\item If $\mathrm{supp}(\mu^{\boxplus n/k})\subset (0,\infty)$, then, almost surely as $d \to \infty$, $\Phi_{\mu,d}$ is $k$-positive.
\item If $\mathrm{supp}(\mu^{\boxplus n/k})\cap (-\infty ,0)\neq \emptyset$, then, almost surely as $d \to \infty$, $\Phi_{\mu,d}$ is not $k$-positive.
\end{enumerate}
\end{theorem}

\begin{proof}
Let $C_d$ be the Choi matrix of the map $\Phi_{\mu,d}$, i.e.~ $C_d = U_dX_dU_d^*$ for a sequence of real diagonal matrices $(X_d)_d$ converging in distribution to $\mu$.
The assumption implies that there exists some positive constant $\delta >0$ such that $\mathrm{supp}(\mu^{\boxplus n/k})\subset [\delta,\infty)$. Fix $M > \max \mathrm{supp}(\mu)$ and let $\varepsilon>0$ be a small constant which will be fixed later. Fix an $\varepsilon$-net $\mathcal Q$ of projections of rank $k$ on $\mathcal{M}_n(\mathbb{C})$ (for the operator norm).
This net can be chosen to be finite because the collection of projections is a compact space. In other words, for any rank $k$ projection $P \in \mathcal M_n(\mathbb C)$, there exists a rank $k$ projection $Q \in \mathcal Q$ such that $\|P-Q\| \leq \varepsilon$. 

For a fixed element $Q \in \mathcal Q$, we obtain, using Theorem \ref{thm:RMT-norm} and Lemma \ref{lem:nica-speicher}, almost surely as $d \to \infty$,
$$\lim_{d \to \infty} \lambda_{\min}\left[\left.(Q \otimes I_d \cdot C_d \cdot Q \otimes I_d)\right|_{\mathrm{ran}(Q) \otimes \mathbb C^d}\right] = \frac{k}{n} \min \mathrm{supp}(\mu^{\boxplus n/k})\geq \frac{k \delta}{n}.$$
Taking the intersection of a finite number of almost sure events, we obtain the same conclusion for all elements $Q$ of $\mathcal Q$. Moreover, since, for all rank $k$ projections $P$, we have almost surely
$$\|P \otimes I_d \cdot C_d \cdot P \otimes I_d - Q \otimes I_d \cdot C_d \cdot Q \otimes I_d\| \leq \|P-Q\| \cdot \|C_d\| \leq M\varepsilon,$$
we conclude that, almost surely as $d \to \infty$,
$$\lim_{d \to \infty} \min_{P \text{ projection of rank $k$}} \lambda_{\min}\left[\left.(P \otimes I_d \cdot C_d \cdot P \otimes I_d)\right|_{\mathrm{ran}(P) \otimes \mathbb C^d}\right] \geq  \frac{k \delta}{n} - M \varepsilon,$$
quantity which is positive as soon as $\varepsilon < (k\delta)/(nM)$. By Proposition \ref{prop:k-positive-from-Choi},
this completes the first part.

To prove the second point, consider any fixed orthogonal projection $P \in \mathcal M_n(\mathbb C)$ of rank $k$. Using Theorem \ref{thm:RMT-norm}, we obtain that, almost surely as $d \to \infty$, the smallest non-trivial eigenvalue of the matrix $(P \otimes I_d) C_d (P \otimes I_d)$ converges to $k/n$ times the bottom of the support of $\mu^{\boxplus n/k}$, which is a negative constant. Therefore, by Proposition \ref{prop:k-positive-from-Choi}, the map $\Phi_{\mu,d}$ is not positive, almost surely as $d \to \infty$.
\end{proof}

When comparing the result above with Theorem \ref{thm:k-positivity-from-mu}, the reader will notice that one needs a gap around zero in the random matrix formulation. This is due to the fact that one needs to avoid eigenvalues fluctuating around zero in order to draw a conclusion about the positivity of a matrix. Finally, from the above theorem, one can immediately deduce the following corollary, simply by replacing the almost sure notion of convergence of random variables with the weaker one of convergence in probability. 

\begin{corollary}
The sequence of random linear maps $(\Phi_{\mu,d})_d$ has the following properties:
\begin{enumerate}
\item If $\mathrm{supp}(\mu^{\boxplus n/k})\subset (0,\infty)$, then
$$\lim_{d \to \infty} \mathbb P\left[ \Phi_{\mu,d} \text{ is $k$-positive} \right] = 1.$$
\item If $\mathrm{supp}(\mu^{\boxplus n/k})\cap (-\infty ,0)\neq \emptyset$, then
$$\lim_{d \to \infty} \mathbb P\left[ \Phi_{\mu,d} \text{ is $k$-positive} \right] = 0.$$
\end{enumerate}
\end{corollary}

\section{GUE quantum states and maps}
\label{sec:GUE}

In this section, we focus on (possibly traceless) random matrices from the Gaussian Unitary Ensemble (GUE). The second chapter of the monograph \cite{agz} is an excellent exposition of the subject. The interest of GUE random matrices is two-fold. Firstly, it is the case where the norm behavior of random matrices is best understood. Secondly, GUE matrices are both unitarily invariant (their distribution is invariant with respect to unitary conjugations) and Wigner matrices (the entries of the random matrix are independent and identically distributed, up to symmetries). We write $\mathrm{i} = \sqrt{-1}$ in order to avoid confusions with the index $i$ of matrix elements. 

\begin{definition}\label{def:GUE}
A random matrix $Z \in \mathcal M_d(\mathbb C)$ is said to have the $\mathrm{GUE}_d$ distribution if its entries are as follows:
$$Z_{ij} = \begin{cases}
X_{ii}/\sqrt{d}, \quad &\text{if } i=j\\
(X_{ij} + \mathrm{i} Y_{ij})/\sqrt{2d}, \quad &\text{if } i<j\\
\bar Z_{ji}, \quad &\text{if } i>j,
\end{cases}$$
where $\{X_{ij},Y_{ij}\}_{i,j=1}^d$ are i.i.d.~ centered, standard real Gaussian random variables. 
\end{definition}

We recall next the Wigner theorem, as well as the stronger, norm convergence of GUE matrices. 

\begin{theorem}
Consider a sequence of $\mathrm{GUE}_d$ random matrices $X_d \in \mathcal M_d(\mathbb C)$. Then, almost surely as $d \to \infty$, the empirical distribution of $X_d$ converges weakly to a standard semicircular distribution: 
$$\frac{1}{d} \sum_{i=1}^d \lambda_i(X_d) \to \mathrm{SC}_{0,1}.$$
Moreover, the norm of $X_d$ converges almost surely to the maximum of the support:
$$\lim_{d \to \infty} \|X_d\| = 2.$$
\end{theorem}

We start with the following simple but important lemma, which will allow us to construct different realizations of GUE matrices acting on a tensor product space. Below, that the matrices $\{E_{ij}\}$ denote the matrix units of $\mathcal M_n(\mathbb C)$. 

\begin{lemma}\label{lem:GUE-n-d-construction}
For any integers $n,d$, consider a family of i.i.d.~ $\mathrm{GUE}_d$ random matrices $\{X_{ij},Y_{ij}\}_{i,j=1}^n$. Define
\begin{align*}
\mathcal M_{nd}(\mathbb C) \ni Z &:= \frac{1}{\sqrt n}\sum_{i=1}^n E_{ii}\otimes X_{ii} \\
&\qquad + \frac{1}{\sqrt{2n}}\sum_{1 \leq i < j \leq n} E_{ij}\otimes (X_{ij} + \mathrm{i}Y_{ij}) \\
&\qquad + \frac{1}{\sqrt{2n}}\sum_{1 \leq i < j \leq n} E_{ji}\otimes (X_{ij} - \mathrm{i}Y_{ij}).
\end{align*}
Then, $Z$ has a 
$\mathrm{GUE}_{nd}$ distribution.
\end{lemma}
\begin{proof}
The $(nd)^2$ elements of $Z$ are centered, Gaussian, independent (up to symmetries) random variables. The normalizations are chosen in such a way that the variances are exactly those of a $\mathrm{GUE}_{nd}$ random matrix.
\end{proof}

\subsection{PPT entanglement}

We are now going to construct an asymptotically positive semidefinite and PPT matrix $Z_d \in \mathcal M_n(\mathbb C) \otimes \mathcal M_d(\mathbb C)$ as well as a map $\Phi_d$ which will detect the entanglement in $Z_d$. Both constructions are inspired from Lemma \ref{lem:GUE-n-d-construction}.

Consider a fixed integer $n \geq 1$ and fixed parameters $\varepsilon >0$ and $\alpha \in (0,1)$. Let $\{X_{ij}\}_{1\leq i \leq j \leq n}$ and $\{Y_{ij}\}_{1\leq i < j \leq n}$ be two families of i.i.d.~ $\mathrm{GUE}_d$ random matrices. Define the matrices 
\begin{align}
\nonumber \mathcal M_{nd}(\mathbb C) \ni Z_d :=2I_n \otimes I_d - \alpha &\left[ \frac{1}{\sqrt n}\sum_{i=1}^n E_{ii}\otimes X_{ii} \right. \\
\label{eq:Z-d} &+ \frac{1}{\sqrt{2n}}\sum_{1 \leq i < j \leq n} E_{ij}\otimes (X_{ij} + \mathrm{i}Y_{ij}) \\
\nonumber & \left. + \frac{1}{\sqrt{2n}}\sum_{1 \leq i<j \leq n} E_{ji}\otimes (X_{ij} - \mathrm{i}Y_{ij})\right],
\end{align}
and
\begin{align}
\nonumber \mathcal M_{nd}(\mathbb C) \ni C_d :=\frac{2+ \varepsilon}{\sqrt n} I_n \otimes I_d &+  \frac{1}{\sqrt n}\sum_{i=1}^n E_{ii}\otimes X_{ii}^\top \\
\label{eq:C-d} &+ \frac{1}{\sqrt{2n}}\sum_{1 \leq i < j \leq n} E_{ij}\otimes (X_{ij}^\top - \mathrm{i}Y_{ij}^\top) \\
\nonumber &+ \frac{1}{\sqrt{2n}}\sum_{1 \leq i<j \leq n} E_{ji}\otimes (X_{ij}^\top + \mathrm{i}Y_{ij}^\top),
\end{align}
and consider the linear map $\Phi_d : \mathcal M_n(\mathbb C) \to \mathcal M_d(\mathbb C)$ whose Choi matrix is $C_d$. In other words, for an input $\rho \in M_n(\mathbb C)$, 
\begin{equation}\label{eq:Phi-d}
\Phi_d(\rho) = \sum_{i,j=1}^n \rho_{ij}C_d(i,j),\end{equation}
where $C_d(i,j) \in \mathcal M_d(\mathbb C)$ denotes the $(i,j)$ block of $C_d$. 

\begin{theorem}\label{thm:GUE-PPT-entangled}
For the linear maps defined in \eqref{eq:Z-d},\eqref{eq:C-d}, and \eqref{eq:Phi-d}, the following statements hold:
\begin{enumerate}
\item For all values of $n, \alpha$, almost surely as $d \to \infty$, the random matrix $Z_d$ is positive semidefinite and PPT.
\item For all values of $n, \varepsilon$, almost surely as $d \to \infty$, the linear map $\Phi_d$ is positive.
\item If $B_d$ denotes the normalized Bell state \eqref{eq:B-d}, then, almost surely as $d \to \infty$, 
$$\lim_{d \to \infty} \langle B_d, [\Phi_d \otimes \mathrm{id}_d](Z_d) B_d \rangle =  2(2+\varepsilon)\sqrt n - \alpha n.$$
In particular, for any choice of parameters such that $2(2+\varepsilon) < \alpha \sqrt n$, the matrix $Z_d$ is, almost surely as $d \to \infty$, asymptotically \emph{PPT} and \emph{entangled}.
\end{enumerate}
\end{theorem}
\begin{proof}
For the first point, note that the matrix $Z_d$ can be written as $Z_d = 2I_{nd}-\alpha \tilde Z_d$, where $\tilde Z_d$ is a $\mathrm{GUE}_{nd}$ matrix (see Lemma \ref{lem:GUE-n-d-construction}. It follows that the support of its asymptotic spectral distribution is the interval $[2-2\alpha, 2+2\alpha]$, which is supported on the positive half-line; strong convergence of $Z_d$ proves that $Z_d$ is asymptotically positive semidefinite. It is also PPT, since transposed $\mathrm{GUE}$ matrices are still $\mathrm{GUE}$.

For the second claim, to show positivity of the map $\Phi_d$ having $C_d$ as a Choi matrix, we are going to use Theorem \ref{thm:k-positivity-from-mu-random-matrix}. Indeed, at fixed $n$, as $d \to \infty$, the random matrix $C_d$ has asymptotic spectral distribution $\mu=\mathrm{SC}_{(2+\varepsilon)/\sqrt n,1}$. We get $\mu^{\boxplus n} = \mathrm{SC}_{(2+\varepsilon)\sqrt n,\sqrt n}$, which has support $[\varepsilon\sqrt n,(4+\varepsilon)\sqrt n] \subset (0, \infty)$, hence $\Phi_d$ is asymptotically positive, by Theorem \ref{thm:k-positivity-from-mu-random-matrix}.

Let us now show that the third point holds. Denoting by $A(i,j) \in \mathcal M_d(\mathbb C)$ the $(i,j)$ block of a matrix $A \in \mathcal M_{nd}(\mathbb C)$, a direct linear algebra computation shows that
\begin{align*} \langle B_d, [\Phi_d \otimes \mathrm{id}_d](Z_d) B_d \rangle &= \frac{1}{d} \sum_{i,j=1}^n \mathrm{Tr}[Z_d(i,j) C_d(i,j)^\top]\\
&=\frac{1}{d} \left[  \frac{2(2+\varepsilon)}{\sqrt n} nd - \frac{\alpha}{n} \sum_{i=1}^n \mathrm{Tr}(X_{ii}^2)  \right. \\
&\left. \qquad-\frac{\alpha}{n}\sum_{1 \leq i < j \leq n} \mathrm{Tr}\left[ (X_{ij}+\mathrm{i}Y_{ij}) (X_{ij}-\mathrm{i}Y_{ij})\right]  \right].
\end{align*}
Using, for all $i,j$, almost surely as $d \to \infty$
$$\lim_{d \to \infty} \frac 1 d \mathrm{Tr}(X_{ij}^2) = \lim_{d \to \infty} \frac 1 d \mathrm{Tr}(Y_{ij}^2) = 1$$ 
and
$$\lim_{d \to \infty} \frac 1 d \mathrm{Tr}(X_{ij}Y_{ij}) = \lim_{d \to \infty} \frac 1 d \mathrm{Tr}(X_{ij})  =  \lim_{d \to \infty} \frac 1 d \mathrm{Tr}(Y_{ij}) =  0,$$ 
we can conclude.
\end{proof}

The construction of the matrix $Z_d$ can be interpreted as follows. Let us replace the notation $Z_d$ by $Z_d(\alpha )$, i.e. view  the Gaussian component of $Z_d$ as
a `direction' and $\alpha$ as an intensity when choosing at random from the 
vicinity of the maximally mixed state (everything up to trivial scalars). With probability $1$ we have that as $d\to\infty$, $Z_d(1)$ is a multiple of a PPT state, and that actually
it is on the boundary of PPT in the sense that if $\alpha = 1$ is replaced by a greater value of
$\alpha$, it will leave the convex cone of PPT matrices. 

The following corollary is itself of interest:
\begin{corollary}
For any $n>16$, there exist $\alpha, \varepsilon$ such that the map $\Phi_d$ detects the entanglement in the PPT matrix $Z_d$. 
 \end{corollary}

In other words, as soon as $n>16$, the maps we introduced in sections \ref{sec:free-prob-version} and \ref{sec:RMT} can detect PPT states that are not separable. The analysis above also shows that the maps $\Phi_d$ become very efficient as entanglement detection tools for $n$ large. Indeed, 

\begin{corollary}\label{cor:threshold-detection}
For any $1>\alpha > 4/\sqrt{n}$ the positive matrix $Z_d$ is PPT entangled.
\end{corollary}

This following result is a direct extension of Theorem \ref{thm:GUE-PPT-entangled}. 

\begin{proposition}
Fix parameters $n,l$ such that $n>16 l$. For any $1>\alpha > 4\sqrt{l}/\sqrt{n}$, the positive matrix $Z_d$ from \eqref{eq:Z-d} is PPT but not $l$-separable,  with high probability as $d\to\infty$.
\end{proposition}
\begin{proof}
The proof goes exactly as for Theorem \ref{thm:GUE-PPT-entangled}, except that one has to replace positivity by $l$-positivity, which affects the threshold accordingly. 
\end{proof}

\subsection{Random separability}

Finally, we would like to finish this section by mentioning that the threshold $4/\sqrt{n}$ from Corollary \ref{cor:threshold-detection}, although perhaps
not optimal, is of the right order with respect to $n$, when $n$ is large.

\begin{theorem}
Let $X_d$ be a standard $\mathrm{GUE}$ random matrix in $\mathcal M_{nd}(\mathbb C)$, and $\alpha \in (0,1)$ be a fixed constant. With probability $1$ as $d\to\infty$, the matrix $2 I_{nd}+\alpha X_d$ is separable as soon as $\frac 2 \alpha >  2 + \frac{4(n-1)}{ \sqrt n}$. In particular, if $\alpha <\sqrt{n}/(2(n-1)+\sqrt n)$, the hermitian matrix $Z_d$ defined in \eqref{eq:Z-d} is, almost surely as $d \to \infty$, separable.
\end{theorem}

\begin{proof}
Consider a fanily $\{X_{ij}^{(1)},X_{ij}^{(2)},X_{ij}^{(3)},X_{ij}^{(4)}\}_{1\leq i \leq j \leq n}$ of i.i.d.~ $\mathrm{GUE}_d$ operators. For $\beta \in (0,1)$, define, for $1\leq i<j \leq n$, the operators 
\begin{align*}
T_{ij}^{(1)} &= (E_{ii}+E_{ij}+E_{ji}+E_{jj})\otimes (\beta X_{ij}^{(1)} +2 I_d) \\
T_{ij}^{(2)} &= (E_{ii}-E_{ij}-E_{ji}+E_{jj})\otimes (\beta X_{ij}^{(2)} +2 I_d) \\
T_{ij}^{(3)} &= (E_{ii}+\mathrm{i}E_{ij}-\mathrm{i}E_{ji}+E_{jj})\otimes (\beta X_{ij}^{(3)} +2 I_d) \\
T_{ij}^{(4)} &= (E_{ii}-\mathrm{i}E_{ij}+\mathrm{i}E_{ji}+E_{jj})\otimes (\beta X_{ij}^{(4)} +2 I_d).
\end{align*}

For $1\leq i\leq n$ we further introduce 
$$\tilde T_i=E_{ii}\otimes \left[-\beta\sum_{j\neq i}\left( X_{ij}^{(1)}+X_{ij}^{(2)}+X_{ij}^{(3)}+X_{ij}^{(4)} \right) +2\beta X_{ii}^{(1)} - 8(n-1)I_d  \right],$$
where we write, for convenience, $X^{(\cdot)}_{ji} = X^{(\cdot)}_{ij}$ for $i<j$. Let
$$\tilde Y_d = \frac{1}{2\beta \sqrt n} \left[ \sum_{1 \leq i < j \leq n} \sum_{s=1}^4 T_{ij}^{(s)} + \sum_{i=1}^n \tilde T_i \right].$$
A direct computations shows that $\tilde Y_d$ follows a $\mathrm{GUE}_{nd}$ distribution. Define 
\begin{equation}\label{eq:Y-d-SEP}
Y_d := x I_{nd} + \tilde Y_d =  \frac{1}{2\beta \sqrt n} \left[ \sum_{1 \leq i < j \leq n} \sum_{s=1}^4 T_{ij}^{(s)} + \sum_{i=1}^n T_i \right],
\end{equation}
where the new diagonal blocks read 
$$ T_i=E_{ii}\otimes \left[-\beta\sum_{j\neq i}\left( X_{ij}^{(1)}+X_{ij}^{(2)}+X_{ij}^{(3)}+X_{ij}^{(4)} \right) +2\beta X_{ii}^{(1)} - 8(n-1)I_d + 2 x \beta \sqrt n I_d \right].$$
We look now for the range of $x$ such that equation \eqref{eq:Y-d-SEP} is a separable decomposition for the matrix $Y_d$. Since $\beta <1$, each of the operators $\beta X_{ij}^{(1)} +2 I_d$ is, almost surely as $d \to \infty$, positive semidefinite; hence, the same hols for all $T_{ij}^{(\cdot)}$. Moreover, for all $i$, the limiting eigenvalue distribution of the random matrix $T_i$ is
$$\mathrm{SC}_{2 x \beta \sqrt n - 8(n-1), 2 \beta \sqrt n},$$
which has positive support if and only if 
$$x \geq 2 + \frac{4(n-1)}{\beta \sqrt n}.$$
Optimizing over the values of $\beta \in (0,1)$, we conclude that $xI_{nd}+\mathrm{GUE}_{nd}$ is separable whenever 
$$x > 2 + \frac{4(n-1)}{\sqrt n}.$$
For the matrix $Z_d$ in \eqref{eq:Z-d}, we have now that
$$Z_d = \alpha\left( \frac 2 \alpha + \mathrm{GUE}_{nd} \right)$$
is separable whenever 
$$\frac 2 \alpha >  2 + \frac{4(n-1)}{ \sqrt n},$$
which finishes the proof.

\end{proof}

\begin{corollary}
In almost all directions, for any $n>16$, the convex body of separable states is contained between two euclidean balls of 
relative radii $4(\sqrt n + 2(n-1))/n$. For $n$ large, the ratio that we can obtain is of the order of $8$. 
\end{corollary}

Note that the comparison is quite interesting with \cite[Corollary 3]{gba} that says that the largest euclidean ball 
inside separable sets in our above setup is of radius $1/n$. Obviously our order of $1/\sqrt{n}$ is much larger, 
however, our result is true 'only' in 'almost all' directions, so it is actually compatible -- and of intrinsic interest.

Our results also complement the following finding of Aubrun, Szarek and Ye in \cite{asy}, which can be stated informally as follows:

\begin{theorem}
Let $\rho \in \mathcal M_{nm}(\mathbb C)$ be a random quantum state distributed along the induced measure with parameters $(nm, s)$. Then, there exist constants $c_{1,2}$ and a function $s_0=s_0(n,m)$ such that
$$ c_1 nm \min(n,m) \leq s_0 \leq c_2 nm \log(nm)^2 \min(n,m)$$
and 
\begin{enumerate}
\item If $s < (1 - \varepsilon)s_0$, the probability that $\rho$ is separable is exponentially low. 
\item If $s > (1 + \varepsilon)s_0$, the probability that $\rho$ is entangled is exponentially low. 
\end{enumerate}
\end{theorem}

\begin{remark}
In the vein of the above theorem, it would be interesting to find good examples of typically $l$-separable
non-centered GUE. Our construction above is specific to the one-separable case, and we do not see 
how to extend it in a non-trivial way to the $l$-separable setup.
\end{remark}

\subsection{Indecomposability}

We recall \cite[Definition 1.2.8]{sto} that a positive map $\Phi$ is called \emph{decomposable} if it can be written as
$$\Phi = \Phi_1 + \top \circ \Phi_2,$$
where $\top$ is the transposition map and $\Phi_{1,2}$ are completely positive maps. 

\begin{theorem}
In the setting of Theorem \ref{thm:GUE-PPT-entangled}, as soon as $2(2+\varepsilon) < \alpha \sqrt n$,  almost surely as $d \to \infty$, the positive linear map $\Phi_d$ having Choi matrix $C_d$, is indecomposable.
\end{theorem}
\begin{proof}
This follows from the third point in Theorem \ref{thm:GUE-PPT-entangled}: the map $\Phi_d$ detects (for $d$ large) the entanglement in the PPT state $Z_d$, hence it cannot be decomposable, see \cite[Corollary 7.2.5 and Proposition 7.3.6]{sto}.
\end{proof}

\end{document}